\documentclass[a4paper, UKenglish]{article}
 
\usepackage{a4wide}
\usepackage{cite}
\usepackage{color}
\usepackage{amsmath}
\usepackage{amsthm, amsfonts}
\usepackage{tikz}
\usetikzlibrary{decorations.pathreplacing}
\usepackage[ruled,vlined, linesnumbered]{algorithm2e}
\usepackage{mathtools}
\usepackage{booktabs}
\newcommand{\ra}[1]{\renewcommand{\arraystretch}{#1}}
\usepackage{array}
\usepackage{adjustbox}

\renewcommand{\epsilon}{\varepsilon}

\newcommand{\ceil}[1]{\lceil #1 \rceil}

\newcommand{\OPT}{\mbox{OPT}}

\newcommand{\calS}{\ensuremath{\mathcal{S}}}
\newcommand{\calM}{\ensuremath{\mathcal{M}}}

\newcommand{\Ex}{\ensuremath{\mathbb{E}}}

\newif\ifnotes
\notesfalse
\notestrue
\ifnotes
\newcommand{\bojana}[1]{\textcolor{blue}{{\footnotesize #1}}\marginpar{\raggedright\tiny \textcolor{blue}{Bojana}}}

\definecolor{green}{rgb}{0,0.7,0}
\newcommand{\martin}[1]{\textcolor{green}{{\footnotesize #1}}\marginpar{\raggedright\tiny \textcolor{green}{Martin}}}

\else
\newcommand{\bojana}[1]{}
\newcommand{\martin}[1]{}
\fi

\makeatletter
\DeclareRobustCommand{\rvdots}{%
  \vbox{
    \baselineskip4\p@\lineskiplimit\z@
    \kern-\p@
    \hbox{.}\hbox{.}\hbox{.}
  }}
\makeatother

\title{Combinatorial Secretary Problems with Ordinal Information\footnote{This work was supported by DFG Cluster of Excellence MMCI at Saarland University.}}

\author{Martin Hoefer\thanks{Institute of Computer Science, Goethe University Frankfurt/Main, Germany. {\tt mhoefer@cs.uni-frankfurt.de}} \and 
Bojana Kodric\thanks{Max Planck Institute for Informatics and Saarland University, Saarland Informatics Campus, Germany. {\tt bojana.kodric@mpi-ing.mpg.de}}}

\date{}

\newtheorem{theorem}{Theorem} 
\newtheorem{lemma}{Lemma}
\newtheorem{definition}{Definition}

\begin{document}

\maketitle
\begin{abstract}
The secretary problem is a classic model for online decision making. Recently, combinatorial extensions such as matroid or matching secretary problems have become an important tool to study algorithmic problems in dynamic markets. Here the decision maker must know the numerical value of each arriving element, which can be a demanding informational assumption. In this paper, we initiate the study of combinatorial secretary problems with \emph{ordinal information}, in which the decision maker only needs to be aware of a preference order consistent with the values of arrived elements. The goal is to design online algorithms with small competitive ratios.

For a variety of combinatorial problems, such as bipartite matching, general packing LPs, and independent set with bounded local independence number, we design new algorithms that obtain constant competitive ratios. 
For the matroid secretary problem, we observe that many existing algorithms for special matroid structures maintain their competitive ratios even in the ordinal model. In these cases, the restriction to ordinal information does not represent any additional obstacle. Moreover, we show that ordinal variants of the submodular matroid secretary problems can be solved using algorithms for the linear versions by extending~\cite{FeldmanZ15}. In contrast, we provide a lower bound of $\Omega(\sqrt{n}/(\log n))$ for algorithms that are oblivious to the matroid structure, where $n$ is the total number of elements. This contrasts an upper bound of $O(\log n)$ in the cardinal model, and it shows that the technique of thresholding is not sufficient for good algorithms in the ordinal model.

%

\end{abstract}


\section{Introduction}\label{sec:introduction}

The secretary problem is a classic approach to model online decision making under uncertain input. The interpretation is that a firm needs to hire a secretary. There are $n$ candidates and they arrive sequentially in random order for an interview. Following an interview, the firm learns the value of the candidate, and it has to make an immediate decision about hiring him before seeing the next candidate(s). If the candidate is hired, the process is over. Otherwise, a rejected candidate cannot be hired at a later point in time. The optimal algorithm is a simple greedy rule that rejects all candidates in an initial learning phase. In the following acceptance phase, it hires the first candidate that is the best among all the ones seen so far. It manages to hire the best candidate with optimal probability $1/e$. Notably, it only needs to know if a candidate is the best seen so far, but no exact numerical values.

Since its introduction~\cite{Dynkin63}, the secretary problem has attracted a huge amount of research interest. Recently, a variety of combinatorial extensions have been studied in the computer science literature~\cite{BabaioffIKK08}, capturing a variety of fundamental online allocation problems in networks and markets, such as network design~\cite{KorulaP09}, resource allocation~\cite{KesselheimRTV13}, medium access in networks~\cite{GoebelHKSV14}, or competitive admission processes~\cite{ChenHKLM15}. Prominently, in the matroid secretary problem~\cite{BabaioffIK07}, the elements of a weighted matroid arrive in uniform random order (e.g., weighted edges of an undirected graph $G$). The goal is to select a max-weight independent set of the matroid (e.g., a max-weight forest of $G$). The popular \emph{matroid secretary conjecture} claims that for all matroids, there exists an algorithm with a constant \emph{competitive ratio}, i.e., the expected total weight of the solution computed by the algorithm is at least a constant fraction of the total weight of the optimum solution. Despite much progress on special cases, the conjecture remains open. Beyond matroids, online algorithms for a variety of combinatorial secretary problems with downward-closed structure have recently been studied (e.g., matching~\cite{KorulaP09,KesselheimRTV13}, independent set~\cite{GoebelHKSV14}, linear packing problems~\cite{KesselheimRTV14} or submodular versions~\cite{FeldmanZ15,KesselheimT17}).

The best known algorithms for matroid or matching secretary problems rely heavily on knowing the exact weight structure of elements. They either compute max-weight solutions to guide the admission process or rely on advanced bucketing techniques to group elements based on their weight. For a decision maker, in many applications it can be quite difficult to determine an exact cardinal preference for each of the incoming candidates. In contrast, in the original problem, the optimal algorithm only needs \emph{ordinal information} about the candidates. This property provides a much more robust guarantee, since the numerical values can be arbitrary, as long as they are consistent with the preference order. 

In this paper, we study algorithms for combinatorial secretary problems that rely only on ordinal information. We assume that there is an unknown value for each element, but our algorithms only have access to the \emph{total order} of the elements arrived so far, which is consistent with their values. We term this the \emph{ordinal model}; as opposed to the \emph{cardinal model}, in which the algorithm learns the exact values. We show bounds on the competitive ratio, i.e., we compare the quality of the computed solutions to the optima in terms of the exact underyling but unknown numerical values. Consequently, competitive ratios for our algorithms are robust guarantees against uncertainty in the input. Our approach follows a recent line of research by studying the potential of algorithms with ordinal information to approximate optima based on numerical values~\cite{AnshelevichS16,AnshelevichP16,AbrahamIKM07,ChakrabartyS14}.

\subsection{Our Contribution}\label{subsec:contribution}

We first point out that many algorithms proposed in the literature continue to work in the ordinal model. In particular, a wide variety of algorithms for variants of the matroid secretary problem with constant competitive ratios continue to obtain their guarantees in the ordinal model (see Table~\ref{table:known-algorithms} for an overview). This shows that many results in the literature are much stronger, since the algorithms require significantly less information. Notably, the algorithm of~\cite{BateniHZ13} extends to the ordinal model and gives a ratio of $O(\log^2 r)$ for general matroids
, where $r$ is the rank of the matroid. In contrast, the improved algorithms with ratios of $O(\log r)$ and $O(\log \log r)$~\cite{BabaioffIK07,Lachish14,FeldmanSZ15} are not applicable in the ordinal model. 

For several combinatorial secretary problems we obtain new algorithms for the ordinal model. For online bipartite matching we give an algorithm that is $2e$-competitive. We also extend this result to online packing LPs with at most $d$ non-zero entries per variable. Here we obtain an $O(d^{(B+1)/B})$-competitive algorithm, where $B$ is a tightness parameter of the constraints. Another extension is matching in general graphs, for which we give a $8.78$-competitive algorithm.


We give an $O(\alpha_1^2)$-competitive algorithm for the online weighted independent set problem in graphs, where $\alpha_1$ is the local independence number of the graph. For example, for the prominent case of unit-disk graphs, $\alpha_1=5$ and we obtain a constant-competitive algorithm. 

For matroids, we extend a result of~\cite{FeldmanZ15} to the ordinal model: The reduction from submodular to linear matroid secretary can be done with ordinal information for marginal weights of the elements. More specifically, we show that whenever there is an algorithm that solves the matroid secretary problem in the ordinal model on some matroid class and has a competitive ratio of $\alpha$, there is also an algorithm for the submodular matroid secretary problem in the ordinal model on the same matroid class with a competitive ratio of $O(\alpha^2)$. The ratio can be shown to be better if the linear algorithm satisfies some further properties.

Lastly, we consider the importance of knowing the weights, ordering, and structure of the domain. For algorithms that have complete ordinal information but cannot learn the specific matroid structure, we show a lower bound of $\Omega(\sqrt{n}/(\log n))$, even for partition matroids, where $n$ is the number of elements in the ground set. This bound contrasts the $O(\log^2 r)$-competitive algorithm and indicates that learning the matroid structure is crucial in the ordinal model. Moreover, it contrasts the cardinal model, where thresholding algorithms yield $O(\log r)$-competitive algorithms without learning the matroid structure.

For structural reasons, we present our results in a slightly different order. We first discuss the matroid results in Section~\ref{sec:matroids}. Then we proceed with matching and packing in Section~\ref{sec:matchingPacking} and independent set in Section~\ref{sec:independentSet}. All missing proofs are deferred to the Appendix.


\section{Preliminaries and Related Work}
\label{sec:preliminaries}


In the typical problem we study, there is a set $E$ of elements arriving sequentially in random order. The algorithm knows $n = |E|$ in advance. It must accept or reject an element before seeing the next element(s). There is a set $\calS \subseteq 2^E$ of \emph{feasible solutions}. $\calS$ is downward-closed, i.e., if $S \in \calS$, then $S' \in \calS$ for every $S' \subseteq S$. The goal is to accept a feasible solution that maximizes an objective function $f$. In the \emph{linear version}, each element has a \emph{value} or \emph{weight} $w_e$, and $f(S) = \sum_{e \in S} w_e$. In the \emph{submodular version}, $f$ is submodular and $f(\emptyset) = 0$. 

In the linear ordinal model, the algorithm only sees a strict total order over the elements seen so far that is consistent with their weights (ties are broken arbitrarily). For the submodular version, we interpret the value of an element as its marginal contribution to a set of elements. In this case, our algorithm has access to an \emph{ordinal oracle} $\mathcal{O}(S)$. For every subset $S$ of arrived elements, $\mathcal{O}(S)$ returns a total order of arrived elements consistent with their marginal values $f(e\vert S) = f(S\cup \{e\}) - f(S)$. 

Given this information, we strive to design algorithms that will have a small competitive ratio $f(S^*)/\Ex[f(S_{alg})]$. Here $S^*$ is an optimal feasible solution and $S_{alg}$ the solution returned by the algorithm. Note that $S_{alg}$ is a random variable due to random-order arrival and possible internal randomization of the algorithm.

In the matroid secretary problem, the pair $\calM = (E,\calS)$ is a matroid. We summarize in Table~\ref{table:known-algorithms} some of the existing results for classes of the (linear) problem that transfer to the ordinal model. The algorithms for all restricted matroid classes other than the graphic matroid assume a-priori complete knowledge of the matroid -- only weights are revealed online. The algorithms do not use cardinal information, their decisions are based only on ordinal information. As such, they translate directly to the ordinal model. Notably, the algorithm from~\cite{BateniHZ13} solves even the general submodular matroid secretary problem in the ordinal model. 

%

\begin{table}[t]
\centering
\ra{1.5}
\begin{tabular}{@{}l||ccccccc@{}} \hline
Matroid & general & k-uniform & graphic & cographic & transversal & laminar & regular \\ \hline
Ratio & $O(\log^2 r)$ & $1+O(\sqrt{1/k})$, $e$ & $2e$ & $3e$ & $16$ & $3\sqrt{3}e$ & $9e$ \\ \hline
Reference & \cite{BateniHZ13} & \cite{Dynkin63,Kleinberg05,BabaioffIKK07} & \cite{KorulaP09} & \cite{Soto13} & \cite{DimitrovP12} & \cite{JailletSZ13} & \cite{DinitzK14} \\ \hline
\end{tabular}

\caption{Existing algorithms for matroid secretary problems that provide the same guarantee in the ordinal model.}
\label{table:known-algorithms}
\end{table}




\subsection{Related Work}\label{subsec:related}
Our work is partly inspired by~\cite{AnshelevichS16,AnshelevichS16WINE}, who study ordinal approximation algorithms for classical optimization problems. They design constant-factor approximation algorithms for matching and clustering problems with ordinal information and extend the results to truthful mechanisms. Our approach here differs due to online arrival. Anshelevich et al.~\cite{AnshelevichP16} examine the quality of randomized social choice mechanisms when agents have metric preferences but only ordinal information is available to the mechanism. Previously,~\cite{AbrahamIKM07,ChakrabartyS14} studied ordinal measures of efficiency in matchings, for instance the average rank of an agent's partner.

The literature on the secretary problem is too broad to survey here. We only discuss directly related work on online algorithms for combinatorial variants. Cardinal versions of these problems have many important applications in ad-auctions and item allocation in online markets~\cite{HajiaghayiKP04}. For multiple-choice secretary, where we can select any $k$ candidates, there are algorithms with ratios that are constant and asymptotically decreasing in $k$~\cite{Kleinberg05,BabaioffIKK07}. More generally, the matroid secretary problem has attracted a large amount of research interest~\cite{BabaioffIK07,ChakrabortyL12,Lachish14,FeldmanSZ15}, and the best-known algorithm in the cardinal model has ratio $O(\log \log r)$. For results on specific matroid classes, see the overview in Table~\ref{table:known-algorithms}. Extensions to the submodular version are treated in~\cite{BateniHZ13,FeldmanZ15}.

Another prominent domain is online bipartite matching, in which one side of the graph is known in advance and the other arrives online in random order, each vertex revealing all incident weighted edges when it arrives~\cite{KorulaP09}. In this case, there is an optimal algorithm with ratio $e$~\cite{KesselheimRTV13}. Moreover, our paper is related to G\"obel et al.~\cite{GoebelHKSV14} who study secretary versions of maximum independent set in graphs with bounded inductive independence number $\rho$. They derive an $O(\rho^2)$-competitive algorithm for unweighted and an $O(\rho^2 \log n)$-competitive algorithm for weighted independent set.

In addition, algorithms have been proposed for further variants of the secretary problem, e.g., the temp secretary problem (candidates hired for a fixed duration)~\cite{FiatGKN15}, parallel secretary (candidates interviewed in parallel)~\cite{FeldmanT12}, or local secretary (several firms and limited feedback)~\cite{ChenHKLM15}. For these variants, some existing algorithms (e.g., for the temp secretary problem in~\cite{FiatGKN15}) directly extend to the ordinal model. In general, however, the restriction to ordinal information poses an interesting challenge for future work in these domains.

\section{Matroids}\label{sec:matroids}

\subsection{Submodular Matroids}\label{subsec:submodular}

\begin{algorithm}[b]
\caption{Greedy~\cite{FeldmanZ15}}\label{greedy}
\SetKwInOut{Input}{Input}\SetKwInOut{Output}{Output}
\Input{ground set $E$}
\Output{independent set $M$}
\BlankLine
Let $M\leftarrow \emptyset$ and $E'\leftarrow E$.\\
\While{$E'\neq \emptyset$}{Let $u \leftarrow \max_{u'} f(u'\vert M)$ and $E'\leftarrow E'\setminus\{u\}$\;
	 \textbf{if} ($M\cup\{u\}$ independent in $\mathcal{M}$) $\wedge$ ($f(u\vert M)\ge 0$) \textbf{then} add $u$ to $M$\;
}
\end{algorithm}

\begin{algorithm}[t]
\caption{Online(p) algorithm~\cite{FeldmanZ15}}\label{alg:submodreduction}
\SetKwInOut{Input}{Input}\SetKwInOut{Output}{Output}
\Input{$n=\lvert E\rvert$, size of the ground set}
\Output{independent set $Q\cap N$}
Choose $X$ from the binomial distribution $B(n, 1/2)$.\\
Reject the first $X$ elements of the input. Let $L$ be the set of these elements.\\

Let $M$ be the output of Greedy on the set $L$.\\
Let $N\leftarrow \emptyset$.\\
\For{each element $u\in E\setminus L$}{Let $w(u)\leftarrow 0$.\\
\If{$u$ accepted by Greedy applied to $M\cup\{u\}$}{
With probability $p$ do the following:\\
Add $u$ to $N$.\\ 
Let $M_u\subseteq M$ be the solution of Greedy immediately before it adds $u$ to it.\\ 
$w(u)\leftarrow f(u\vert M_u)$.}

Pass $u$ to Linear with weight $w(u)$.}
\Return $Q\cap N$, where $Q$ is the output of Linear.
\end{algorithm}

We start our analysis by showing that -- in addition to algorithms for special cases mentioned above -- a powerful technique for submodular matroid secretary problems~\cite{FeldmanZ15} can be adjusted to work even in the ordinal model. More formally, in this section we show that there is a reduction from submodular matroid secretary problems with ordinal information (SMSPO) to linear matroid secretary problems with ordinal information (MSPO). The reduction uses Greedy (Algorithm~\ref{greedy}) as a subroutine and interprets the marginal value when added to the greedy solution as the value of an element. These values are then forwarded to whichever algorithm (termed Linear) that solves the linear version of the problem. In the ordinal model, we are unable to see the exact marginal values. Nevertheless, we manage to construct a suitable ordering for the forwarded elements. Consequently, we can apply algorithm Linear as a subroutine to obtain a good solution for the ordinal submodular problem.

Let $\calM = (E,\calS)$ be the matroid, $f$ the submodular function, and $E$ the ground set of elements. The marginal contribution of element $u$ to set $M$ is denoted by $f(u \vert M) = f(M \cup \{u\})-f(M)$. Since $f$ can be non-monotone, Greedy in the cardinal model also checks if the marginal value of the currently best element is positive. While we cannot explicitly make this check in the ordinal model, note that $f(u\vert M)\ge 0 \iff f(M\cup \{u\})\ge f(M) = f(M \cup \{u'\})$ for every $u' \in M$. Since the ordinal oracle includes the elements of $M$ in the ordering of marginal values, there is a way to check positivity even in the ordinal model. Therefore, our results also apply to non-monotone functions $f$.

A potential problem with Algorithm~\ref{alg:submodreduction} is that we must compare marginal contributions of different elements w.r.t.\ different sets. We can resolve this issue by following the steps of the Greedy subroutine that tries to add new elements to the greedy solution computed on the sample. We use this information to construct a correct ordering over the marginal contributions of elements that we forward to Linear.

\begin{lemma}\label{lem:structureSubmodular}
Let us denote by $s_u$ the step of Greedy in which the element $u$ is accepted when applied to $M+u$. Then $s_{u_1}<s_{u_2}$ implies $f(u_1\vert M_{u_1}) \ge f(u_2\vert M_{u_2})$.
\end{lemma}

\begin{proof}
First, note that $M_{u_1}\subset M_{u_2}$ when $s_1>s_2$. We denote by $m_{u_1}$ the element of $M$ that would be taken in step $s_{u_1}$ if $u_1$ would not be available. Then we know that $f(u_1\vert M_{u_1})\ge f(m_{u_1}\vert M_{u_1})$. Furthermore, since $s_1<s_2$, $f(m_{u_1}\vert M_{u_1})\ge f(u_2\vert M_{u_1})$. Lastly, by using submodularity, we know that $f(u_2\vert M_{u_1}) \ge f(u_2\vert M_{u_2})$.
\end{proof}

When $s_{u_1}=s_{u_2}$, then $M_{u_1}=M_{u_2}$ so the oracle provides the order of marginal values. Otherwise, the lemma yields the ordinal information. Thus, we can construct an ordering for the elements that are forwarded to Linear that is consistent with their marginal values in the cardinal model. Hence, the reduction can be applied in the ordinal model, and all results from~\cite{FeldmanZ15} continue to hold. We mention only the main theorem. It implies constant ratios for all problems in Table~\ref{table:known-algorithms} in the submodular version. 

\begin{theorem}
Given an arbitrary algorithm Linear for MSPO that is $\alpha$-competitive on a matroid class, there is an algorithm for SMSPO with competitive ratio is at most $24\alpha(3\alpha+1)=O(\alpha^2)$ on the same matroid class. For SMSPO with monotone $f$, it can be improved to $8\alpha(\alpha+1)$.
\label{thm:submodular}
\end{theorem}



\subsection{A Lower Bound}
\label{subsec:lower-bound}

Another powerful technique in the cardinal model is thresholding, where we first sample a constant fraction of the elements to learn their weights. Based on the largest weight observed, we pick a threshold and accept subsequent elements greedily if they exceed the threshold. This approach generalizes the classic algorithm~\cite{Dynkin63} and provides logarithmic ratios for many combinatorial domains~\cite{BabaioffIK07,KorulaP09,GoebelHKSV14,ChenHKLM15}. Intuitively, these algorithms \emph{learn the weights but not the structure}.

We show that this technique does not easily generalize to the ordinal model. The algorithms with small ratios in the ordinal model rely heavily on the matroid structure. Indeed, in the ordinal model we show a polynomial lower bound for algorithms in the matroid secretary problem that learn the ordering but not the structure. Formally, we slightly simplify the setting as follows. The algorithm receives the global ordering of all elements in advance. It determines (possibly at random) a threshold position in the ordering. Then elements arrive and are accepted greedily if ranked above the threshold. Note that the algorithm does not use sampling, since in this case the only meaningful purpose of sampling is learning the structure. We call this a \emph{structure-oblivious} algorithm.

\begin{theorem}\label{thm:lower}
Every structure-oblivious randomized algorithm has a competitive ratio of at least $\Omega(\sqrt{n}/(\log n))$.
\label{thm:random-lowerbound}
\end{theorem}

\begin{proof}
In the proof, we restrict our attention to instances with weights in $\{0,1\}$ (for a formal justification, see Lemma~\ref{lem:01} in the Appendix). We give a distribution of such instances on which every deterministic algorithm has a competitive ratio of $\Omega(\sqrt{n}/(\log n))$. Using Yao's principle, this shows the claimed result for randomized algorithms.

All instances in the distribution are based on a graphic matroid (in fact, a partition matroid) of the following form. There is a simple path of $1+k$ segments. The edges in each segment have weight of $0$ or $1$. We call the edges with value $1$ in the last $k$ segments the ``valuable edges''. The total number of edges is the same in each instance and equals $n+1$. All edges in the first segment have value $1$ and there is exactly one edge of value $1$ in all other segments (that being the aforementioned valuable edges). In the first instance there are in total $k+1$ edges of value $1$ (meaning that there is only one edge in the first segment). In each of the following instances this number is increased by $k$ (in the $i$-th instance there are $(i-1)\cdot k+1$ edges in the first segment) such that the last instance has only edges with value $1$ (there are $n-k+1$ edges in the first segment). The zero edges are always equally distributed on the last $k$ path segments. The valuable edges are lower in the ordering than any non-valuable edge with value $1$ (see Figure~\ref{fig:instance-values}). Each of the instances appears with equal probability of $\frac{k}{n}$ (see Figure~\ref{fig:threshold-alg-instance} for one example instance). 
\begin{figure}[t]
{\centering
\begin{minipage}{.48\textwidth}
\hspace*{-0.5em}{
\begin{tikzpicture}[scale=0.66, every node/.style={transform shape}]
\draw (0,0) -- (10,0);
\draw (-0.1,0) -- (0,0);
\draw[line width=1mm] (0,0) -- (0.5,0);
\draw (-0.1,-0.2) -- (-0.1,0.2);
\draw (10,-0.2) -- (10, 0.2) node[xshift=0cm, yshift=0.27cm]{$n$};
\draw[thick] (9.5,-0.1) -- (9.5, 0.1) node[xshift=-.1cm, yshift=0.4cm]{$n\hspace{-1mm}-\hspace{-1mm}k$};
\draw[thick] (9,-0.1) -- (9, 0.1);
\draw[thick] (8.5,-0.1) -- (8.5, 0.1);
\draw[thick] (8,-0.1) -- (8, 0.1);
\draw[thick] (7.5,-0.1) -- (7.5, 0.1);
\draw[thick] (7,-0.1) -- (7, 0.1);
\draw[thick] (6.5,-0.1) -- (6.5, 0.1);
\draw[thick] (6,-0.1) -- (6, 0.1);
\draw[thick] (5.5,-0.1) -- (5.5, 0.1);
\draw[thick] (5,-0.1) -- (5, 0.1);
\draw[thick] (4.5,-0.1) -- (4.5, 0.1);
\draw[thick] (4,-0.1) -- (4, 0.1);
\draw[thick] (3.5,-0.1) -- (3.5, 0.1);
\draw[thick] (3,-0.1) -- (3, 0.1);
\draw[thick] (2.5,-0.1) -- (2.5, 0.1);
\draw[thick] (2,-0.1) -- (2, 0.1);
\draw[thick] (1.5,-0.1) -- (1.5, 0.1) node[midway, yshift=0.4cm]{$3k$};
\draw[thick] (1,-0.1) -- (1, 0.1) node[midway, yshift=0.4cm]{$2k$};
\draw[thick] (0.5,-0.1) -- (0.5, 0.1) node[midway, yshift=0.4cm]{$k$};
\node [] (13) [xshift=2.5cm, yshift=0.4cm] {\dots};
\draw (0,-0.1) -- (0,0.1) node[midway, yshift=0.4cm]{$1$};
\draw[decorate, decoration={brace, amplitude=5pt, mirror}] (0.6,-0.2) -- (9.9,-0.2) node[midway,yshift=-0.6cm]{$0$};
\draw[decorate, decoration={brace, amplitude=2pt, mirror}] (-0.1,-0.3) -- (0.5,-0.3) node[midway,yshift=-0.5cm]{$1$};

\draw (-0.1,-1.5) -- (10,-1.5);
\draw[line width=1mm] (0.5, -1.5) -- (1, -1.5);
\draw (-0.1,-1.7) -- (-0.1,-1.3);
\draw[thick] (0,-1.6) -- (0,-1.4);
\draw (10,-1.7) -- (10,-1.3);
\draw[thick] (9.5,-1.6) -- (9.5, -1.4);
\draw[thick] (9,-1.6) -- (9, -1.4);
\draw[thick] (8.5,-1.6) -- (8.5, -1.4);
\draw[thick] (8,-1.6) -- (8, -1.4);
\draw[thick] (7.5,-1.6) -- (7.5, -1.4);
\draw[thick] (7,-1.6) -- (7, -1.4);
\draw[thick] (6.5,-1.6) -- (6.5, -1.4);
\draw[thick] (6,-1.6) -- (6, -1.4);
\draw[thick] (5.5,-1.6) -- (5.5, -1.4);
\draw[thick] (5,-1.6) -- (5, -1.4);
\draw[thick] (4.5,-1.6) -- (4.5, -1.4);
\draw[thick] (4,-1.6) -- (4, -1.4);
\draw[thick] (3.5,-1.6) -- (3.5, -1.4);
\draw[thick] (3,-1.6) -- (3, -1.4);
\draw[thick] (2.5,-1.6) -- (2.5, -1.4);
\draw[thick] (2,-1.6) -- (2, -1.4);
\draw[thick] (1.5,-1.6) -- (1.5, -1.4);
\draw[thick] (1,-1.6) -- (1, -1.4);
\draw[thick] (0.5,-1.6) -- (0.5, -1.4);
\draw[decorate, decoration={brace, amplitude=5pt, mirror}] (1.1,-1.7) -- (9.9,-1.7) node[midway,yshift=-0.8cm]{$0$};
\draw[decorate, decoration={brace, amplitude=2pt, mirror}] (-0.1,-1.8) -- (1,-1.8) node[midway,yshift=-0.65cm]{$1$};

\node [] (1) [yshift=-2.5cm, xshift=5cm] {\rvdots};

\draw (-0.1,-3.2) -- (10,-3.2);
\draw[line width=1mm] (9.5,-3.2) -- (10,-3.2);
\draw[thick] (0,-3.3) -- (0,-3.1);
\draw (10,-3.3) -- (10,-3.1);
\draw (-0.1,-3.4) -- (-0.1,-3);
\draw (10,-3.4) -- (10, -3);
\draw[thick] (9.5,-3.3) -- (9.5, -3.1);
\draw[thick] (9,-3.3) -- (9, -3.1);
\draw[thick] (8.5,-3.3) -- (8.5, -3.1);
\draw[thick] (8,-3.3) -- (8, -3.1);
\draw[thick] (7.5,-3.3) -- (7.5, -3.1);
\draw[thick] (7,-3.3) -- (7, -3.1);
\draw[thick] (6.5,-3.3) -- (6.5, -3.1);
\draw[thick] (6,-3.3) -- (6, -3.1);
\draw[thick] (5.5,-3.3) -- (5.5, -3.1);
\draw[thick] (5,-3.3) -- (5, -3.1);
\draw[thick] (4.5,-3.3) -- (4.5, -3.1);
\draw[thick] (4,-3.3) -- (4, -3.1);
\draw[thick] (3.5,-3.3) -- (3.5, -3.1);
\draw[thick] (3,-3.3) -- (3, -3.1);
\draw[thick] (2.5,-3.3) -- (2.5, -3.1);
\draw[thick] (2,-3.3) -- (2, -3.1);
\draw[thick] (1.5,-3.3) -- (1.5, -3.1);
\draw[thick] (1,-3.3) -- (1, -3.1);
\draw[thick] (0.5,-3.3) -- (0.5, -3.1);
\draw[decorate, decoration={brace, amplitude=5pt, mirror}] (-0.1,-3.5) -- (10,-3.5) node[midway,yshift=-0.65cm]{$1$};
\end{tikzpicture}
\caption{Values for the family of instances described in the proof of Theorem~11, where the position of the ``valuable'' ones is denoted by the thick segment. \label{fig:instance-values}}}
\end{minipage}
\hfill
\begin{minipage}{.49\textwidth}
\vspace{-8.81mm}
\adjustbox{trim={0} {.18\height} {0} {0\height}, clip}{
\begin{tikzpicture}[auto, node distance = 2cm, main node/.style={circle, inner sep=2.5pt, fill, font=\sffamily\Large\bfseries}, scale=0.61, every node/.style={transform shape}, baseline={(10,-2.7)}]
\node[main node] (1) {};
\node[main node] (2) [right of=1] {};
\node[main node] (3) [right of=2] {};
\node[main node] (4) [right of=3] {};
\node[main node] (5) [xshift=9cm] {};
\node[main node] (6) [right of=5] {};
\node [] (7) [yshift=-1.5cm, xshift=1cm] {\rvdots};
\node [] (8) [yshift=1.5cm, xshift=3cm] {\rvdots};
\node [] (9) [yshift=1.5cm, xshift=5cm] {\rvdots};
\node [] (10) [yshift=1.5cm, xshift=10cm] {\rvdots};
\node [] (11) [xshift=6.7cm] {};
\node [] (12) [xshift=8.3cm] {};
\node [] (13) [xshift=7.5cm] {\dots};

\path[]
  (1) edge node[below, yshift=0.1cm] {\scriptsize 1} (2)
      edge [bend right = 30] node[below, yshift=0.1cm] {\scriptsize 1} (2)
      edge [bend right = 70] node[below] {\scriptsize 1} (2)
      edge [out=270, in=270, looseness=3] node[below] {\scriptsize 1} (2)
  (2) edge node[below, yshift=0.02cm] {\scriptsize 1} (3)
    edge [bend left = 30] node[above] {\scriptsize 0} (3)
    edge [bend left = 70] node[above] {\scriptsize 0} (3)
      edge [out=-270, in=-270, looseness=3] node[above] {\scriptsize 0} (3)
  (3) edge node[below, yshift=0.1cm] {\scriptsize 1} (4)
    edge [bend left = 30] node[above] {\scriptsize 0} (4)
    edge [bend left = 70] node[above] {\scriptsize 0} (4)
      edge [out=-270, in=-270, looseness=3] node[above] {\scriptsize 0} (4)
  (4) edge node {} (11)    
  (5) edge node[below, yshift=0.1cm] {\scriptsize 1} (6)
    edge [bend left = 30] node[above] {\scriptsize 0} (6)
    edge [bend left = 70] node[above] {\scriptsize 0} (6)
      edge [out=-270, in=-270, looseness=3] node[above] {\scriptsize 0} (6)
  (12) edge node {} (5)    
  ;
\draw[decorate, decoration={brace, amplitude=5pt}] (10.9,-0.2) -- (2.1,-0.2) node[midway,yshift=-0.3cm]{$k$};
\end{tikzpicture}}
\caption{One instance from the family described in the proof of Theorem~11.\label{fig:threshold-alg-instance}}
\end{minipage}
}
\end{figure}

A deterministic algorithm picks a threshold at position $i$. The expected value of the solution is
\begin{align*}
\mathbb{E}[w(S_{alg})] &\le 1 + \frac{k}{n} \sum_{\ell=1}^{\frac{i}{k}} \frac{k}{\ell}\le 1 + \frac{k^2}{n} \log \frac{i}{k}\le \frac{k^2}{n}\log \frac{n}{k} + 1\enspace,
\end{align*} 
where $\log$ denotes the natural logarithm and the expression results from observing that the algorithm cannot obtain more than a value of $1$ if its threshold $i$ falls above the valuable $1$'s. Otherwise it gets an additional fraction of $k$, depending on how close the threshold is positioned to the valuable $1$'s. For instance, if the threshold is set between $1$ and $k$ positions below the valuable $1$'s, the algorithm will in expectation select edges of total value of at least $1 + k/2$. This follows from the random arrival order of the edges and the fact that the ratio of valuable to non-valuable edges that the algorithm is ready to accept is at least $1:2$.
Furthermore, we see that for this distribution of instances the optimal way to set a deterministic threshold is at the lowest position. Using $k=\sqrt{n}$, a lower bound on the competitive ratio is
\[
\frac{k}{\frac{k^2}{n}\log \frac{n}{k} + 1} = \frac{n}{k\log \frac{n}{k} + \frac{n}{k}} = \Omega\left(\frac{\sqrt{n}}{\log n}\right)\enspace.\qedhere
\]
\end{proof}


\section{Matching and Packing}
\label{sec:matchingPacking}


\subsection{Bipartite Matching}\label{subsec:bipartite-matching}

In this section, we study online bipartite matching. The vertices on the right side of the graph (denoted by $R$) are static and given in advance. The vertices on the left side (denoted by $L$) arrive sequentially in a random order. 
Every edge $e = (r,\ell) \in R \times L$ has a non-negative weight $w(e) \ge 0$. In the cardinal model, each vertex of $L$ reveals upon arrival the weights of all incident edges. In the ordinal model, we are given a total order on all edges that have arrived so far, consistent with their weights. Before seeing the next vertex of $L$, the algorithm has to decide to which vertex $r \in R$ (if any) it wants to match the current vertex $\ell$. A match that is formed cannot be revoked. The goal is to maximize the total weight of the matching.

The algorithm for the cardinal model in~\cite{KesselheimRTV13} achieves an optimal competitive ratio of $e$. However, this algorithm heavily exploits cardinal information by repeatedly computing max-weight matchings for the edges seen so far. For the ordinal model, our Algorithm~\ref{alg:bipartite-ordinal} below obtains a competitive ratio of $2e$. While similar in spirit, the main difference is that we rely on a greedy matching algorithm, which is based solely on ordinal information. It deteriorates the ratio only by a factor of 2.

Here we assume to have access to ordinal preferences over all the edges in the graph. Note that the same approach works if the vertices provide correlated (ordinal) preference lists consistent with the edge weights, for every vertex from $R$ and every arrived vertex from $L$. In this case, the greedy algorithm can still be implemented by iteratively matching and removing a pair that mutually prefers each other the most, and it provides an approximation guarantee of 2 for the max-weight matching (see, e.g.,~\cite{AnshelevichH12}). In contrast, if we receive only preference lists for vertices on one side, there are simple examples that establish super-constant lower bounds on the competitive ratio\footnote{Consider a bipartite graph with two nodes on each side (named A,B and 1,2). If we only know that both A and B prefer 1 to 2, the ratio becomes at least 2 even in the offline case. Similar examples imply that the (offline) ratio must grow in the size of the graph.}. 

\begin{algorithm}[t]
\caption{Bipartite Matching}\label{alg:bipartite-ordinal}
\SetKwInOut{Input}{Input}\SetKwInOut{Output}{Output}

\Input{vertex set $R$ and cardinality $n=|L|$}
\Output{matching $M$}
\BlankLine
Let $L'$ be the first $\lfloor \frac{n}{e} \rfloor$ vertices of $L$, and $M\leftarrow \emptyset$\;
\For{each $\ell\in L\setminus L'$}{
	$L'\leftarrow L'\cup \{\ell\}$\;
  $M^{(\ell)}\leftarrow \text{greedy matching on } G[L'\cup R]$\; 	
	Let $e^{(\ell)}\leftarrow (\ell,r)$ be the edge assigned to $\ell$ in $M^{(\ell)}$\;
	\textbf{if} $M\cup \{e^{(\ell)}\}$ is a matching \textbf{then} add $e^{(\ell)}$ to $M$\;
}

\end{algorithm}

\begin{lemma}\label{bipartite-matching-lemma}
Let the random variable $A_v$ denote the contribution of the vertex $v\in L$ to the output, i.e. weight assigned to $v$ in $M$. Let $w(M^*)$ denote the value of the maximum-weight matching in $G$. For $\ell\in \{\ceil{\frac{n}{e} },\dots,n\}$,

$$\mathbb{E}\big[A_{\ell}\big]\geq \frac{\lfloor {\frac{n}{e}} \rfloor}{\ell-1}\cdot \frac{w(M^*)}{2n}\enspace.
$$
\end{lemma}
\begin{proof}
We first show that $e^{(\ell)}$ has a significant expected weight. Then we bound the probability of adding $e^{(\ell)}$ to $M$.

In step $\ell$, $|L'|=\ell$ and the algorithm computes a greedy matching $M^{(\ell)}$ on $G[L'\cup R]$. The current vertex $\ell$ can be seen as selected uniformly at random from $L'$, and $L'$ can be seen as selected uniformly at random from $L$. Therefore, $\mathbb{E}[w(M^{(\ell)})]\geq \frac{\ell}{n}\cdot \frac{w(M^*)}{2}$ and $\mathbb{E}[w(e^{(\ell)})]\geq \frac{w(M^*)}{2n}$. Here we use that a greedy matching approximates the optimum by at most a factor of $2$~\cite{AnshelevichH12}.

Edge $e^{(\ell)}$ can be added to $M$ if $r$ has not been matched already. The vertex $r$ can be matched only when it is in $M^{(k)}$. The probability of $r$ being matched in step $k$ is at most $\frac{1}{k}$ and the order of the vertices in steps $1,\dots, k-1$ is irrelevant for this event.

\begin{align*}
\Pr[r \text{ unmatched in step } \ell] &= \Pr\Bigg[\bigwedge_{k=\ceil{n/e}}^{\ell-1} r\not\in e^{(k)}\Bigg]\geq \prod_{k=\ceil{n/e}}^{\ell-1} \frac{k-1}{k} = \frac{\ceil{\frac{n}{e}}-1}{\ell - 1}
\end{align*}

We now know that $\Pr[M\cup e^{(\ell)} \text{ is a matching}]\geq \frac{\lfloor n/e\rfloor}{\ell-1}$. Using this and $\mathbb{E}[w(e^{(\ell)})]\geq\frac{w(M^*)}{2n}$, the lemma follows. 
\end{proof}

\begin{theorem}
Algorithm~\ref{alg:bipartite-ordinal} for bipartite matching is $2e$-competitive. 
\end{theorem}
\begin{proof}
The weight of matching $M$ can be obtained by summing over random variables $A_{\ell}$.
\begin{align*}
\mathbb{E}[w(M)]&=\mathbb{E}\Bigg[\sum_{\ell=1}^n A_{\ell}\Bigg] \geq \sum_{\ell=\ceil{n/e}}^n \frac{\lfloor n/e \rfloor}{\ell-1}\cdot \frac{w(M^*)}{2n}= \frac{\lfloor n/e\rfloor}{2n} \sum_{\ell=\lfloor n/e \rfloor}^{n-1} \frac{1}{\ell}\cdot w(M^*) 
\end{align*}

Since $\frac{\lfloor n/e\rfloor}{n}\geq \frac{1}{e} - \frac{1}{n}$ and $\sum_{\ell=\lfloor n/e\rfloor}^{n-1} \frac{1}{\ell}\geq \ln \frac{n}{\lfloor n/e\rfloor}\geq 1$, it follows that $$\mathbb{E}[w(M)]\geq \bigg(\frac{1}{e}-\frac{1}{n}\bigg)\cdot \frac{w(M^*)}{2}\enspace.$$
\end{proof}

In the submodular version of the offline problem, the natural greedy algorithm gives a $3$-approximation \cite{FisherNW78}. It builds the matching by greedily adding an edge that maximizes the marginal improvement of $f$, which is the information delivered by the ordinal oracle. When using this algorithm as a subroutine for the bipartite matching secretary problem, the resulting procedure achieves a $12$-approximation in the submodular case~\cite{KesselheimT17}.


\subsection{Packing}\label{subsec:packing}

\newcommand{\bfc}{\mathbf{c}}
\newcommand{\bfx}{\mathbf{x}}
\newcommand{\bfy}{\mathbf{y}}
\newcommand{\bfb}{\mathbf{b}}
\newcommand{\bfA}{\mathbf{A}}
\newcommand{\bfZero}{\mathbf{0}}
\newcommand{\bfOne}{\mathbf{1}}

Our results for bipartite matching can be extended to online packing LPs of the form $\max \bfc^{\tau}\bfx$ s.t. $\bfA\bfx\le \bfb$ and $\bfZero \le \bfx \le \bfOne$, which model problems with $m$ resources and $n$ online requests coming in random order. Each resource $i\in [m]$ has a capacity $b_i$ that is known in advance, together with the number of requests. Every online request comes with a set of options, where each option has its profit and resource consumption. Once a request arrives, the coefficients of its variables are revealed and the assignment to the variables has to be determined.

Formally, request $j\in [n]$ corresponds to variables $x_{j,1},\dots, x_{j,K}$ that represent $K$ options. Each option $k\in[K]$ contributes with profit $c_{j,k}\ge0$ and has resource consumption $a_{i,j,k}\ge0$ for resource $i$. Overall, at most one option can be selected, i.e., there is a constraint $\sum_{k\in[K]} x_{j,k}\le1, \forall j\in[n]$. The objective is to maximize total profit while respecting the resource capacities. The offline problem is captured by the following linear program:
\begin{align*}
\max \sum_{j\in [n]}\sum_{k\in[K]} c_{j,k}x_{j,k}\quad \text{ s.t. }\quad &\sum_{j\in[n]}\sum_{k\in[K]} a_{i,j,k}x_{j,k} \le b_i &i\in[m]\\
&\sum_{k\in[K]} x_{j,k} \le 1 &j\in[n]
\end{align*}

As a parameter, we denote by $d$ the maximum number of non-zero entries in any column of the constraint matrix $\bfA$, for which by definition $d\le m$. We compare the solution to the fractional optimum, which we denote by $\bfx^*$. The competitive ratio will be expressed in terms of $d$ and the capacity ratio $B= \min_{i\in[m]}\left \lfloor\frac{b_i}{\max_{j\in[n],k\in[K]}a_{i,j,k}}\right \rfloor$.



Kesselheim et al.~\cite{KesselheimRTV13} propose an algorithm that heavily exploits cardinal information -- it repeatedly solves an LP-relaxation and uses the solution as a probability distribution over the options. Instead, our Algorithm~\ref{alg:packing} for the ordinal model is based on greedy assignments in terms of profits $c_{j,k}$. More specifically, the greedy assignment considers variables $x_{j,k}$ in non-increasing order of $c_{j,k}$. It sets a variable to $1$ if this does not violate the capacity constraints, and to 0 otherwise.

\begin{algorithm}[t]
\caption{Packing LP}\label{alg:packing}
\SetKwInOut{Input}{Input}\SetKwInOut{Output}{Output}

\Input{capacities $\bf{b}$, total number of requests $n$, probability $p=\frac{e(2d)^{1/B}}{1+e(2d)^{1/B}}$}
\Output{assignment vector $\bfy$}
\BlankLine
Let $L'$ be the first $p\cdot n$ requests, and $\bfy \leftarrow \bfZero$\;
\For{each $j \notin L'$}{
	$L'\leftarrow L'\cup \{j\}$\; 
	$\bfx^{(L')}\leftarrow \text{greedy assignment on the LP for } L'$\;
	$\bfy_j \leftarrow \bfx_j^{(L')}$\;
	\textbf{if} $\lnot (\bfA(\bfy) \le \bfb)$ \textbf{then} $\bfy_j \leftarrow \bfZero$\;
}
\end{algorithm}

\begin{theorem}\label{thm:packing}
Algorithm~\ref{alg:packing} for online packing LPs is $O(d^{(B+1)/B})$-competitive.
\end{theorem}


\subsection{Matching in General Graphs}\label{subsec:general-matching}

\begin{algorithm}[b]
\caption{General Matching}\label{alg:matching-ordinal}
\SetKwInOut{Input}{Input}\SetKwInOut{Output}{Output}

\Input{vertex set $V$ and cardinality $n=|V|$}
\Output{matching $M$}
\BlankLine
Let $R$ be the first $\lfloor \frac{n}{2}\rfloor$ vertices of $V$\;
Let $L'$ be the further $\lfloor \frac{n}{2e} \rfloor$ vertices of $V$, and $M\leftarrow \emptyset$\;
\For{each $\ell\in V\setminus L'$}{
	$L'\leftarrow L'\cup \{\ell\}$\;
  $M^{(\ell)}\leftarrow \text{greedy matching on } G[L'\cup R]$\; 
	Let $e^{(\ell)}\leftarrow (\ell,r)$ be the edge assigned to $\ell$ in $M^{(\ell)}$\;
	\textbf{if} $M\cup \{e^{(\ell)}\}$ is a matching \textbf{then} add $e^{(\ell)}$ to $M$\;
}

\end{algorithm}

Here we study the case when vertices of a general undirected graph arrive in random order. In the beginning, we only know the number $n$ of vertices. Each edge in the graph has a non-negative weight $w(e) \ge 0$. Each vertex reveals the incident edges to previously arrived vertices and their weights (cardinal model), or we receive a total order over all edges among arrived vertices that is consistent with the weights (ordinal model). An edge can be added to the matching only in the round in which it is revealed. The goal is to construct a matching with maximum weight.

We can tackle this problem by prolonging the sampling phase and dividing the vertices into ``left'' and ``right'' vertices. Algorithm~\ref{alg:matching-ordinal} first samples $n/2$ vertices. These are assigned to be the set $R$, corresponding to the static side of the graph in bipartite matching. The remaining vertices are assigned to be the set $L$. The algorithm then proceeds by sampling a fraction of the vertices of $L$, forming a set $L'$. The remaining steps are exactly the same as in Algorithm~\ref{alg:bipartite-ordinal}.

\begin{theorem}\label{thm:generalMatching}
Algorithm~\ref{alg:matching-ordinal} for matching in general graphs is $12e/(e+1)$-competitive, where $12e/(e+1) < 8.78$.
\end{theorem}



\section{Independent Set and Local Independence}\label{sec:independentSet}

In this section, we study maximum independent set in graphs with bounded local independence number. The set of elements are the vertices $V$ of an underyling undirected graph $G$. Each vertex has a weight $w_v \ge 0$. We denote by $N(v)$ the set of direct neighbors of vertex $v$. Vertices arrive sequentially in random order and reveal their position in the order of weights of vertices seen so far. The goal is to construct an independent set of $G$ with maximum weight. The exact structure of $G$ is unknown, but we know that $G$ has a bounded local independence number $\alpha_1$.

\begin{definition}
An undirected graph $G$ has local independence number $\alpha_1$ if for each node $v$, the cardinality of every independent set in the neighborhood $N(v)$ is at most $\alpha_1$.
\end{definition}

We propose Algorithm~\ref{alg:sample-and-price}, which is inspired by the Sample-and-Price algorithm for matching in~\cite{KorulaP09}. Note that G\"obel et al.~\cite{GoebelHKSV14} construct a more general approach for graphs with bounded \emph{inductive} independence number $\rho$. However, they only obtain a ratio of $O(\rho^2 \log n)$ for the weighted version, where a competitive ratio of $\Omega(\log n/\log^2 \log n)$ cannot be avoided, even in instances with constant $\rho$. These algorithms rely on $\rho$-approximation algorithms for the offline problem that crucially exploit cardinal information.

\begin{algorithm}[t]
\caption{Independent Set in Graphs with Bounded Local Independence Number}\label{alg:sample-and-price}
\SetKwInOut{Input}{Input}\SetKwInOut{Output}{Output}

\Input{$n=|G|$, $p=\sqrt{\alpha_1/(\alpha_1 + 1)}$}
\Output{independent set of vertices $S$}
\BlankLine
Set $k\leftarrow \text{Binom}(n,p)$, $S\leftarrow\emptyset$\;
Reject first $k$ vertices of $G$, denote this set by $G'$\;
Build a maximal independent set of vertices from $G'$ greedily, denote this set by $M_1$\;
\For{each $v\in G\setminus G'$}{
	$w^* \leftarrow \max \{ w \mid \mathcal{N}(v)\cap M_1\}$\;
	\textbf{if} $(v>w^*) \wedge (S\cup \{v\}$ independent set) \textbf{then} add $v$ to $S$\;
}
\end{algorithm}

\begin{algorithm}[t]
\caption{Simulate }\label{alg:simulate}
\SetKwInOut{Input}{Input}\SetKwInOut{Output}{Output}

\Input{$n=|G|$, $p=\sqrt{\alpha_1/(\alpha_1 + 1)}$}
\Output{independent set of vertices $S$}
\BlankLine
Sort all vertices in $G$ in non-increasing order of value\;
Initialize $M_1, M_2 \leftarrow \emptyset$\;
\For{each $v\in G$ in sorted order}{
	\If{$M_1\cup \{v\}$ independent set}{
	flip a coin with probability $p$ of heads\;
	\textbf{if} heads \textbf{then} $M_1\leftarrow M_1\cup \{v\}$; \textbf{else} $M_2\leftarrow M_2\cup \{v\}$\;
	}
}
$S \leftarrow M_2$\;
\For{each $w\in S$}{
	\textbf{if} $w$ has neighbors in $S$ \textbf{then} remove $w$ and all his neighbors from $S$\;
}

\end{algorithm}

Similar to \cite{KorulaP09}, we reformulate Algorithm~\ref{alg:sample-and-price} into an equivalent approach (Algorithm~\ref{alg:simulate}) for the sake of analysis. Given the same arrival order, the same vertices are in the sample. Algorithm~\ref{alg:simulate} drops all vertices from $S$ that have neighbors in $S$ while Algorithm~\ref{alg:sample-and-price} keeps one of them. Hence, $\Ex[w(S_{Alg_{\ref{alg:sample-and-price}}})]\ge \Ex[w(S_{Sim})]$. In what follows, we analyze the performance of Simulate. The first lemma follows directly from the definition of the local independence number.

\begin{lemma}
$\mathbb{E}[w(M_1)] \ge p\cdot \frac{w(S^*)}{\alpha_1}$, where $\alpha_1 \ge 1$ is the local independence number of $G$.
\end{lemma}

\begin{lemma}\label{lemma:conflicts}
$\Ex\big[\lvert \mathcal{N}(v)\cap M_2\rvert \bigm| v\in M_2\big]\le \frac{\alpha_1 (1-p)}{p}$\enspace.
\end{lemma}

\begin{proof}
Let us denote by $X_u^1$ and $X_u^2$ the indicator variables for the events $u\in M_1$ and $u\in M_2$ respectively. Then,
\begin{align*}
&\Ex\big[\lvert\mathcal{N}(v)\cap M_2\rvert \bigm| v\in M_2\big] = \Ex\Bigg[\sum_{u\in \mathcal{N}(v)} X_u^2 \bigm| v\in M_2\Bigg]= \sum_{u\in \mathcal{N}(v)} \Ex\big[X_u^2 \bigm| v\in M_2\big]\\ 
&= \frac{1-p}{p} \sum_{u\in \mathcal{N}(v)} \Ex\big[X_u^1 \bigm| v\in M_2\big] \le \frac{1-p}{p} \cdot \alpha_1
\end{align*}
\end{proof}

\begin{theorem}
Algorithm~\ref{alg:simulate} for weighted independent set is $O(\alpha_1^2)$-competitive, where $\alpha_1$ is the local independence number of the graph.
\end{theorem}

\begin{proof}
By using Markov's inequality and Lemma~\ref{lemma:conflicts}, 
\[
\begin{array}{llcl}
& \Pr[\lvert\mathcal{N}(v)\cap M_2\rvert \ge 1 \bigm| v\in M_2] &\le& \alpha_1\cdot (1-p)/p \\[0.3cm]
\text{and} & \Pr[\lvert\mathcal{N}(v)\cap M_2\rvert < 1 \bigm| v\in M_2] & > & 1- (\alpha_1(1-p)/p)\enspace.
\end{array}
\]
Thus, we can conclude that 
\begin{align*}
\Ex[w(S)] &\ge \left(1-\alpha_1\cdot\frac{1-p}{p}\right) \cdot \Ex[w(M_2)]\ge \bigg(1-\alpha_1\cdot\frac{1-p}{p}\bigg)\cdot\frac{1-p}{\alpha_1}\cdot w(S^*)\enspace.
\end{align*}
The ratio is optimized for $p=\sqrt{\frac{\alpha_1}{\alpha_1+1}}$, which proves the theorem.
\end{proof}

As a prominent example, $\alpha_1 = 5$ in the popular class of unit-disk graphs. In such graphs, our algorithm yields a constant competitive ratio for online independent set in the ordinal model. 

\bibliography{literature,martin}

\begin{thebibliography}{10}

\bibitem{AbrahamIKM07}
David Abraham, Robert Irving, Telikepalli Kavitha, and Kurt Mehlhorn.
\newblock Popular matchings.
\newblock {\em SIAM J. Comput.}, 37(4):1030--1045, 2007.

\bibitem{AnshelevichH12}
Elliot Anshelevich and Martin Hoefer.
\newblock Contribution games in networks.
\newblock {\em Algorithmica}, 63(1--2):51--90, 2012.

\bibitem{AnshelevichP16}
Elliot Anshelevich and John Postl.
\newblock Randomized social choice functions under metric preferences.
\newblock In {\em Proc.\ 25th Intl.\ Joint Conf.\ Artif.\ Intell.\ (IJCAI)},
  pages 46--59, 2016.

\bibitem{AnshelevichS16}
Elliot Anshelevich and Shreyas Sekar.
\newblock Blind, greedy, and random: Algorithms for matching and clustering
  using only ordinal information.
\newblock In {\em Proc.\ 13th Conf.\ Artificial Intelligence (AAAI)}, pages
  390--396, 2016.

\bibitem{AnshelevichS16WINE}
Elliot Anshelevich and Shreyas Sekar.
\newblock Truthful mechanisms for matching and clustering in an ordinal world.
\newblock In {\em Proc.\ 12th Conf.\ Web and Internet Economics (WINE)}, pages
  265--278, 2016.

\bibitem{BabaioffIKK07}
Moshe Babaioff, Nicole Immorlica, David Kempe, and Robert Kleinberg.
\newblock A knapsack secretary problem with applications.
\newblock In {\em Proc.\ 10th Workshop Approximation Algorithms for
  Combinatorial Optimization Problems (APPROX)}, pages 16--28, 2007.

\bibitem{BabaioffIKK08}
Moshe Babaioff, Nicole Immorlica, David Kempe, and Robert Kleinberg.
\newblock Online auctions and generalized secretary problems.
\newblock {\em SIGecom Exchanges}, 7(2), 2008.

\bibitem{BabaioffIK07}
Moshe Babaioff, Nicole Immorlica, and Robert Kleinberg.
\newblock Matroids, secretary problems, and online mechanisms.
\newblock In {\em Proc.\ 18th Symp.\ Discrete Algorithms (SODA)}, pages
  434--443, 2007.

\bibitem{BateniHZ13}
MohammadHossein Bateni, MohammadTaghi Hajiaghayi, and Morteza Zadimoghaddam.
\newblock Submodular secretary problem and extensions.
\newblock {\em ACM Trans.\ Algorithms}, 9(4):32, 2013.

\bibitem{ChakrabartyS14}
Deeparnab Chakrabarty and Chaitanya Swamy.
\newblock Welfare maximization and truthfulness in mechanism design with
  ordinal preferences.
\newblock In {\em Proc.\ 5th Symp.\ Innovations in Theoret.\ Computer Science
  (ITCS)}, pages 105--120, 2014.

\bibitem{ChakrabortyL12}
Sourav Chakraborty and Oded Lachish.
\newblock Improved competitive ratio for the matroid secretary problem.
\newblock In {\em Proc.\ 23rd Symp.\ Discrete Algorithms (SODA)}, pages
  1702--1712, 2012.

\bibitem{ChenHKLM15}
Ning Chen, Martin Hoefer, Marvin K\"unnemann, Chengyu Lin, and Peihan Miao.
\newblock Secretary markets with local information.
\newblock In {\em Proc.\ 42nd Intl.\ Coll.\ Automata, Languages and Programming
  (ICALP)}, volume~2, pages 552--563, 2015.

\bibitem{DimitrovP12}
Nedialko Dimitrov and Greg Plaxton.
\newblock Competitive weighted matching in transversal matroids.
\newblock {\em Algorithmica}, 62(1-2):333--348, 2012.

\bibitem{DinitzK14}
Michael Dinitz and Guy Kortsarz.
\newblock Matroid secretary for regular and decomposable matroids.
\newblock {\em SIAM J. Comput.}, 43(5):1807--1830, 2014.

\bibitem{Dynkin63}
Eugene Dynkin.
\newblock The optimum choice of the instant for stopping a {M}arkov process.
\newblock In {\em Sov. Math. Dokl}, volume~4, pages 627--629, 1963.

\bibitem{FeldmanSZ15}
Moran Feldman, Ola Svensson, and Rico Zenklusen.
\newblock A simple \emph{O}(log log(rank))-competitive algorithm for the
  matroid secretary problem.
\newblock In {\em Proc.\ 26th Symp.\ Discrete Algorithms (SODA)}, pages
  1189--1201, 2015.

\bibitem{FeldmanT12}
Moran Feldman and Moshe Tennenholtz.
\newblock Interviewing secretaries in parallel.
\newblock In {\em Proc.\ 13th Conf.\ Electronic Commerce (EC)}, pages 550--567,
  2012.

\bibitem{FeldmanZ15}
Moran Feldman and Rico Zenklusen.
\newblock The submodular secretary problem goes linear.
\newblock In {\em Proc.\ 56th Symp.\ Foundations of Computer Science (FOCS)},
  pages 486--505, 2015.

\bibitem{FiatGKN15}
Amos Fiat, Ilia Gorelik, Haim Kaplan, and Slava Novgorodov.
\newblock The temp secretary problem.
\newblock In {\em Proc.\ 23rd European Symp.\ Algorithms (ESA)}, pages
  631--642, 2015.

\bibitem{FisherNW78}
Marshall Fisher, George Nemhauser, and Laurence Wolsey.
\newblock An analysis of approximations for maximizing submodular set
  functions-{II}.
\newblock In {\em Polyhedral combinatorics}, pages 73--87. Springer, 1978.

\bibitem{GoebelHKSV14}
Oliver G\"obel, Martin Hoefer, Thomas Kesselheim, Thomas Schleiden, and
  Berthold V\"ocking.
\newblock Online independent set beyond the worst-case: Secretaries, prophets
  and periods.
\newblock In {\em Proc.\ 41st Intl.\ Coll.\ Automata, Languages and Programming
  (ICALP)}, volume~2, pages 508--519, 2014.

\bibitem{HajiaghayiKP04}
MohammadTaghi Hajiaghayi, Robert Kleinberg, and David Parkes.
\newblock Adaptive limited-supply online auctions.
\newblock In {\em Proc.\ 5th Conf.\ Electronic Commerce (EC)}, pages 71--80,
  2004.

\bibitem{JailletSZ13}
Patrick Jaillet, Jos{\'e} Soto, and Rico Zenklusen.
\newblock Advances on matroid secretary problems: {F}ree order model and
  laminar case.
\newblock In {\em Proc.\ 16th Intl.\ Conf.\ Integer Programming and
  Combinatorial Optimization (IPCO)}, pages 254--265, 2013.

\bibitem{KesselheimRTV13}
Thomas Kesselheim, Klaus Radke, Andreas T\"onnis, and Berthold V\"ocking.
\newblock An optimal online algorithm for weighted bipartite matching and
  extensions to combinatorial auctions.
\newblock In {\em Proc.\ 21st European Symp.\ Algorithms (ESA)}, pages
  589--600, 2013.

\bibitem{KesselheimRTV14}
Thomas Kesselheim, Klaus Radke, Andreas T\"onnis, and Berthold V\"ocking.
\newblock Primal beats dual on online packing {LP}s in the random-order model.
\newblock In {\em Proc.\ 46th Symp.\ Theory of Computing (STOC)}, pages
  303--312, 2014.

\bibitem{KesselheimT17}
Thomas Kesselheim and Andreas T{\"{o}}nnis.
\newblock Submodular secretary problems: Cardinality, matching, and linear
  constraints.
\newblock In {\em Proc.\ 20th Workshop Approximation Algorithms for
  Combinatorial Optimization Problems (APPROX)}, pages 16:1--16:22, 2017.

\bibitem{Kleinberg05}
Robert Kleinberg.
\newblock A multiple-choice secretary algorithm with applications to online
  auctions.
\newblock In {\em Proc.\ 16th Symp.\ Discrete Algorithms (SODA)}, pages
  630--631, 2005.

\bibitem{KorulaP09}
Nitish Korula and Martin P{\'a}l.
\newblock Algorithms for secretary problems on graphs and hypergraphs.
\newblock In {\em Proc.\ 36th Intl.\ Coll.\ Automata, Languages and Programming
  (ICALP)}, pages 508--520, 2009.

\bibitem{Lachish14}
Oded Lachish.
\newblock O(log log rank) competitive ratio for the matroid secretary problem.
\newblock In {\em Proc.\ 55th Symp.\ Foundations of Computer Science (FOCS)},
  pages 326--335, 2014.

\bibitem{Soto13}
Jos{\'e} Soto.
\newblock Matroid secretary problem in the random-assignment model.
\newblock {\em SIAM J. Comput.}, 42(1):178--211, 2013.

\end{thebibliography}

\newpage
\appendix
\section{Appendix}\label{sec:appendix}

\subsection{Proof of Theorem~\ref{thm:packing}}

The proof is based on the following lemma.

\begin{lemma}\label{hypermatching-lemma}
Let the random variable $A_\ell$ denote the contribution of request $\bfx_\ell$ to the output and $\bfc^{\tau} \bfx^*$ the value of the optimal fractional solution. For requests $\bfx_j\in\{pn+1,\dots, n\}$ it holds that
\[
\mathbb{E}[A_j]\ge \Bigg(1-d\cdot\bigg(\frac{e(1-p)}{p}\bigg)^B\Bigg)\frac{\bfc^{\tau}\bfx^*}{(d+1)n}\enspace.
\]
\end{lemma}

\begin{proof}
If $x_{j,k}^{(L')}=1$, then as in proof of Lemma~\ref{bipartite-matching-lemma}, 
we get $\mathbb{E}[\bfc_j\bfx_j^{(L')}]=\mathbb{E}[c_{j,k}x_{j,k}^{(L')}]\ge \frac{\bfc^{\tau}\bfx^*}{(d+1)n}$, where the expectation is taken over the choice of the set $L'$ and the choice of the last vertex in the order of arrival.

The algorithm sets $\bfy_j$ to $\bfx_j^{(L')}$ only if the capacity constraints can be respected. For the sake of analysis, we assume that the algorithm only sets $\bfy_j$ to $\bfx_j^{(L')}$ if every capacity constraint $b_i$ that $\bfx_j^{(L')}$ affects ($x_{j,k}^{(L')}=1$ and $a_{i,j,k}\neq0$) is affected by at most $B-1$ previous requests. We bound the probability of a capacity constraint $b_i$ being affected in any preceding step $s\in\{pn+1,\dots, j-1\}$, for a fixed $i$: 
\begin{align*}
\Pr[b_i\text{ affected by } x_{s,k'}^{(L')}=1]&\le \sum_{\bfx_{j'}\in\{1,\dots,s\}}\Pr[(\bfx_{j'}\text{ is last in the order}) \land (a_{i,j',k'}\neq0)]\\
&\le \frac{1}{s} \sum_{\bfx_{j'}\in\{1,\dots,s\}}\Pr[a_{i,j',k'}\neq0]\le \frac{B}{s}\enspace, 
\end{align*}
where the last step follows from $\mathbf{y}$ being a feasible solution throughout the run of the algorithm.
Now, we bound the probability of not being able to set $\bfy_j$ to $\bfx_j^{(L')}$:
\begin{align*}
\Pr[b_i \text{ is affected at least } B \text{ times}] &\le \sum_{\substack{C\subseteq \{pn+1,\dots,j-1\},\\ \left\vert C \right\vert = B}} \bigg(\prod_{s\in C}\frac{B}{s}\bigg)\\ &\le {(1-p)n \choose B}\cdot \bigg( \frac{B}{pn} \bigg)^B\le \bigg(\frac{(1-p)e}{p} \bigg)^B\enspace,
\end{align*}
so the probability of succeeding in setting $\bfy_j$ to $\bfx_j^{(L')}$ is
\[
\Pr[\bfA\mathbf{y}\le \mathbf{b}] \ge 1 - d\cdot \bigg( \frac{(1-p)e}{p} \bigg)^B\enspace,
\]
because we can do a union bound over all $b_i$ that are affected by $\bfx_j^{(L')}$ and there are at most $d$ such, since that is the maximal number of non-zero entries in any column of the constraint matrix $\bfA$.

Combining this with the inequality regarding the expected contribution of $\bfx_j^{(L')}$, we get the claimed result.
\end{proof}

Using Lemma~\ref{hypermatching-lemma}, we get
\begin{align*}
\mathbb{E}[\bfc^{\tau} \mathbf{y}]&=\sum_{\ell=pn+1}^{n} \mathbb{E}[A_\ell]\ge \sum_{l=pn+1}^{n} \Bigg(1-d\cdot \bigg(\frac{e(1-p)}{p} \bigg)^B\Bigg)\frac{\bfc^{\tau}\mathbf{x}^*}{(d+1)n}\\
&= \frac{\bfc^{\tau}\mathbf{x}^*}{d+1}\cdot \frac{1}{1+e(2d)^{1/B}}\cdot \Bigg(1-d\cdot \bigg(\frac{1}{(2d)^{1/B}} \bigg)^B\Bigg)\ge \frac{\bfc^{\tau}\mathbf{x}^*}{2(d+1)(1+2ed^{1/B})}\enspace.
\end{align*} 
\qed

Note that Theorem~\ref{thm:packing} contains the one-sided b-hypermatching problem as a special case. For the even more special case of $b=1$ in the one-sided hypermatching, an algorithm was given in \cite{KorulaP09}, which also works in the ordinal model. Our ratio in this special case is similar, but our approach extends to arbitrary capacities $b \ge 1$.


\subsection{Proof of Theorem~\ref{thm:generalMatching}}

The proof is based on the following lemma.

\begin{lemma}
Let the random variable $A_\ell$ denote the contribution of vertex $\ell\in{\lfloor n/2 + n/(2e)\rfloor}$ to the output, i.e. the weight of the edge assigned to $\ell$ in $M$. Then,
\[
\mathbb{E}[A_\ell] \ge \frac{\ceil{\frac{n}{2} + \frac{n}{2e}}}{\ell-1} \cdot \frac{1}{2} \cdot \frac{w(M^*)}{n}
\] 
\end{lemma}

\begin{proof}
The proof is similar to the one of Lemma 5, 
with the additional observation that each edge is available with probability $\frac{1}{2}$. It is available only if the incident vertices are assigned to different sides of the bipartition.
\end{proof}

We now use the lemma to bound as follows:
\begin{align*}
\mathbb{E}[w(M)] &= \mathbb{E} \Bigg[\sum_{\ell=1}^n A_\ell\Bigg]\ge \sum_{\ell=\ceil{n/2 + n/(2e)}}^n \frac{\lfloor n/2 + n/(2e)\rfloor}{\ell-1} \cdot \frac{1}{2} \cdot \frac{w(M^*)}{n} \\
&= \frac{\lfloor n/2 + n/(2e)\rfloor}{n} \cdot \frac{w(M^*)}{2} \cdot \sum_{\ell=\lfloor n/2 + n/(2e)\rfloor}^{n-1} \frac{1}{\ell}\\
&\ge \bigg( \frac{1}{2}\Big( 1+\frac{1}{e}\Big) - \frac{1}{n}\bigg) \cdot \frac{w(M^*)}{2} \cdot \frac{1}{3}
\end{align*}
\qed


\subsection{Competitive Ratio and 0--1 Weights}
\label{subsec:0-1-weights}

For worst-case bounds, we can restrict our attention to instances where all elements have cardinal weights in $\{0,1\}$. These instances always result in the worst competitive ratio, as shown in the following lemma.

\begin{lemma}\label{lem:01}
By converting an arbitrary weighted instance to an instance with weights in $\{0,1\}$, the competitive ratio between the optimum solution and the solution computed by an algorithm based on ordinal information can only deteriorate.
\end{lemma}

\begin{proof}
Without loss of generality, we assume that all elements of the original instance have distinct weights. We denote the elements chosen in the optimal solution by $a_1^*,\dots, a_k^*$ and the elements chosen by the algorithm by $b_1, \dots, b_m$. The numbering respects the ordinal ordering of weights, i.e., $a_1^* \succ a_2^*\succ \dots a_k^*$ and $b_1 \succ b_2 \succ \dots \succ b_m$. The competitive ratio is
\[
\frac{\OPT}{\text{ALG}} = \frac{w(a_1^*) + \dots + w(a_k^*)}{w(b_1) + \dots + w(b_m)}\enspace.
\]
This ratio can only increase if we change the weight of all elements that appear after $a_k^*$ in the global ordering to $0$. This effectively shortens the set of elements with a contribution chosen by the algorithm to $b_1,\dots, b_\ell$, for some suitable $\ell \le m$. Furthermore, we change the weights of all elements between $a_i^*$ and $a_{i+1}^*$ by decreasing them to $a_{i+1}^*$. We now denote the elements that the algorithm chose by $c_1,\dots,c_l$, since their weights might have changed. Both of these changes do not influence $\OPT$, but they reduce the weight of the solution returned by the algorithm. We continue converting the instance, by focusing on $w(a_k^*)$. Then,
\begin{align*}
\frac{\OPT}{\text{ALG}}&= \frac{w(a_1^*) + \dots + w(a_k^*)}{w(b_1) + \dots + w(b_m)} \le \frac{w(a_1^*) + \dots + w(a_k^*)}{w(b_1) + \dots + w(b_\ell)}\\&\le \frac{w(a_1^*) + \dots + w(a_k^*)}{w(c_1) + \dots + w(c_\ell)}= \frac{A+w(a_k^*)}{B+r\cdot w(a_k^*)}\enspace, 
\end{align*}
where $A=w(a_1^*)+\dots+w(a_{k-1}^*)$, $B$ is the sum of the weights of all elements that the algorithm chose which are not equal to $w(a_k^*)$ in the altered instance and $r\in \mathbb{N}_0$.

Taking the derivative for $w(a_k^*)$,
\[
\bigg( \frac{A+w(a_k^*)}{B+r\cdot w(a_k^*)} \bigg)' = \frac{B-r\cdot A}{(B+r\cdot w(a_k^*))^2}\enspace,
\]
we either decrease $w(a_k^*)$ to $0$ or raise it to $w(a_{k-1}^*)$ (depending what makes the ratio increase, i.e., the sign of the derivative). We continue this procedure until all weights of the instance are equal to either to $a_1^*$ or $0$. Note that these changes preserve the global ordering. W.l.o.g., we can finally set $w(a_1^*) = 1$. 

Note that we increased the ratio between the solution of the algorithm and the optimal solution for the original weights, when applying the transformed weights. Note that none of these transformations change the decisions of the algorithm. In contrast, the optimum solution for the transformed weights can only become better, which even further deteriorates the competitive ratio.
\end{proof}

\end{document}